\documentclass[letterpaper, 10 pt, conference]{ieeeconf}
\IEEEoverridecommandlockouts
\usepackage{amsmath,amssymb,amsthm}
\usepackage{mathtools}
\usepackage{bm}
\usepackage{bbm}           
\usepackage{graphicx}
\usepackage{subcaption}
\usepackage{cite}
\usepackage{url}
\usepackage{hyperref}
\usepackage{xcolor}
\usepackage{algorithm}
\usepackage{algorithmic}

\usepackage{booktabs}
\usepackage{multirow}
\usepackage{tabularx}
\usepackage{xspace}
\usepackage{makecell}
\usepackage{booktabs}

\newtheorem{theorem}{Theorem}

\newtheorem{remark}{Remark}

\newtheorem{assumption}{Assumption}

\newcommand{\X}{\mathcal{X}}

\newcommand{\K}{\mathcal{K}}
\newcommand{\D}{\mathcal{D}}

\newcommand{\Real}{\mathbb{R}}

\newcommand{\nK}{n_{\mathcal{K}}}
\newcommand{\nO}{n_{\mathcal{O}}}
\newcommand{\STL}{\texorpdfstring{\ensuremath{\mathrm{STL}}}{STL}\xspace}
\newcommand{\BHK}{\texorpdfstring{\ensuremath{\mathrm{BHK}}}{BHK}\xspace}
\newcommand{\GCS}{\texorpdfstring{\ensuremath{\mathrm{GCS}}}{GCS}\xspace}
\newcommand{\norm}[1]{\left\lVert #1 \right\rVert}
\newcommand{\minimize}{\text{minimize}}

\begin{document}

\title{Motion Planning with Temporal Logic Precedence Specifications
in Dynamic Environments via Spacetime Augmented
Graphs of Convex Sets}

\author{Shilin You \quad Lars Lindemann \quad Tyler~H.~Summers%
  \thanks{S. You and T.~H.~Summers are with the Erik Jonsson School of Engineering and
    Computer Science, University of Texas at Dallas. E-mail: \texttt{tyler.summers@utdallas.edu}. L. Lindemann is with the Automatic Control Laboratory, ETH Z\"urich.}%
  \thanks{This work is supported by the United States Air Force Office of
    Scientific Research under Grant FA9550-23-1-0424, and by the National
    Science Foundation under Grant ECCS-2047040.}%
}

\maketitle

\begin{abstract}
We extend the augmented graph of convex sets (\GCS) framework for motion planning with temporal logic precedence specifications to dynamic environments in which obstacles, keys, and doors undergo piecewise-linear translation and/or dilation over time. The main insight is that piecewise-linear motion of polytopic sets in $\Real^d$ lifts exactly to a static polytopic partition of spacetime $\Real^d \times \Real_{\geq 0}$. We show that the main result from the static case carries over to spacetime essentially without modification: a shortest path in the spacetime augmented \GCS provides an exact optimal solution to the motion planning problem with precedence specifications. Moreover, the layered subgraph lattice that encodes signal temporal logic (\STL) precedence specifications retains precise correspondence with the Bellman--Held--Karp (\BHK) algorithm, and the full library of restricted but expressive  \STL\ specifications generalizes to the dynamic setting. Further, the framework allows time-gated keys and doors, collectible or passable only within a temporal window. The geometric side of the framework requires two modifications: a causality constraint (the time coordinate of a trajectory must be non-decreasing) and a spacetime speed constraint (bounded physical velocity), both of which are convex in the B\'{e}zier control points.
\end{abstract}


\section{Introduction}
\label{sec:introduction}

Motion planning in dynamic environments requires joint consideration of geometry and time. The problem becomes substantially more challenging when trajectories must also satisfy logical precedence specifications, such as collecting a key before passing through a door or visiting task regions in an admissible order. In such settings, planning must jointly determine a collision-free continuous trajectory, a feasible ordering of tasks, and the times at which the trajectory visits the corresponding task regions.

The graphs of convex sets (\GCS) framework provides a powerful approach for coupling discrete graph search with continuous trajectory optimization. Marcucci et al.~\cite{marcucci2024shortest} introduced a mixed-integer convex formulation of shortest paths through graphs whose vertices contain convex sets, together with a convex relaxation that is often tight in practice. The framework was subsequently applied to motion planning around obstacles~\cite{marcucci2023motion}, producing smooth, certifiably near-optimal trajectories. Extensions have considered multi-query planning~\cite{morozov2024multiquery}, contact-rich manipulation~\cite{graesdal2024contact}, non-Euclidean configuration spaces~\cite{cohn2025noneuclidean}, unified graph optimization formulations~\cite{marcucci2025unified}, multi-region routing and traveling salesman problem (TSP) variants~\cite{tang2026unified},and methods for improving scalability~\cite{natarajan2024implicit,luna2026augmented}.

\begin{figure}[!t]
    \centering
    \includegraphics[width=0.8\linewidth]{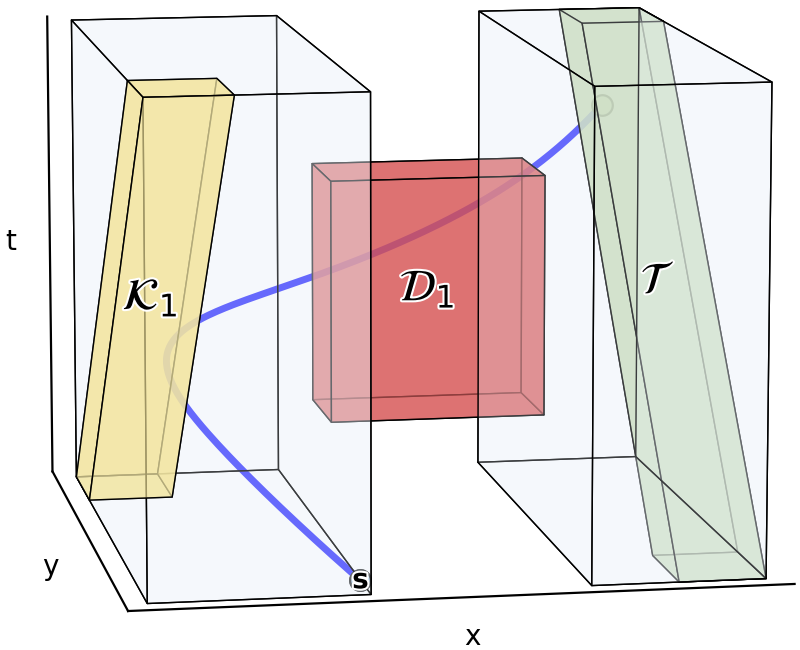}
    \caption{An optimal trajectory in a dynamic 2D environment, lifted to 3D spacetime. A key $K_1(t)$ must be \emph{intercepted} to access the door $D_1(t)$, then arrive at the target region $\mathcal{T}$. Key and target regions are moving, and the door is only accessible during certain time interval.}
    \label{fig:main_demo}
    \vspace{-10pt}
\end{figure}

Temporal-logic specifications have also been incorporated into GCS. Kurtz and Lin~\cite{kurtz2023temporal} construct a product between a temporal-logic automaton and a GCS, providing a general framework for LTL-constrained motion planning. More recently, Chen et al.~\cite{chen2026signal} formulate continuous-time STL motion planning by coupling timed automata with a convex decomposition of the configuration space, yielding a GCS formulation that jointly enforces logical progress, timing, smoothness, and velocity constraints while directly generating B\'ezier trajectories. In contrast, augmented GCS~\cite{you2025augmented,you2026augmented} directly encodes a restricted but expressive fragment of structured precedence specifications in the graph topology. The resulting layered subgraph construction simultaneously selects a feasible task ordering and optimizes the continuous trajectory, while avoiding a general product automaton construction. This yields singly exponential complexity in the number of task regions for the considered specification classes, and~\cite{you2026augmented} extends this framework to a systematic library of eight specification variations.

Several recent works study GCS in spacetime or incorporate kinodynamic constraints in a \emph{vanilla obstacle-avoidance} setting, without temporal logic specifications~\cite{tang2025space,osburn2025systematic}. Other work incorporates signal temporal logic (\STL) specifications into spacetime GCS, but considers static environments and requires a prescribed temporal window for each task region, enforcing task sequences by restricting the admissible free space over time~\cite{de2026stl}. Related optimization-based approaches encode temporal-logic specifications directly through mixed-integer formulations. In particular, Logic Network Flow formulations~\cite{lin2025optimization, lin2025towards} exploit network-flow structure inspired by GCS to obtain tighter convex relaxations and have been applied to systems with piecewise-affine dynamics and real-time replanning in changing
environments.
More broadly, dynamic obstacle avoidance has been extensively studied using sampling-based, model-predictive, reachability-based, and related methods. Our focus here is specifically on the connection between spacetime planning and precedence-specification GCS.

The main distinction from existing spacetime GCS methods is that we consider dynamic environments together with the precedence specifications encoded by augmented GCS. Existing spacetime GCS methods primarily address trajectory planning and obstacle avoidance in dynamic environments, whereas augmented GCS~\cite{you2025augmented,you2026augmented} considers logical task specifications in static environments. In this work, we combine these two capabilities by extending augmented GCS to spacetime, allowing moving obstacles, keys, and doors while preserving the specification structure of the static framework. This extension also enables explicitly time-dependent specifications, such as keys that can only be collected or doors that can only be traversed within prescribed temporal windows.

Our main observation is that a broad class of dynamic environments can be represented exactly using static convex geometry in spacetime. Specifically, a polytope in $\Real^d$ with fixed orientation and piecewise-linear translation and dilation over time lifts, on each linear segment, to a polytope in $\Real^{d+1}$. Consequently, the spacetime free set can be partitioned into convex cells and represented using the same GCS machinery employed in the static setting. This allows the augmented GCS construction and its specification library to be extended directly to dynamic environments.

The main contributions of this paper are:
\begin{enumerate}

    \item \textbf{Spacetime geometry.} We show that with piecewise-linearly translating and dilating polytopic regions, the resulting spacetime freespace set in $\Real^{d+1}$ can be constructed via an exact convex partition algorithm such as reverse search.

    \item \textbf{Spacetime augmented GCS construction and optimality.} We extend the augmented \GCS\ construction to dynamic environments and show that its layered logical structure is preserved in spacetime. In particular, the number and connectivity of logical subgraphs are unchanged; dynamics affect the underlying spacetime geometry rather than the specification structure. We establish that a shortest path in the augmented \GCS\ corresponds to an exact solution of the dynamic motion planning problem.

    \item \textbf{Dynamic and time-gated specifications.} The entire specification library extends to moving task regions, including time-gated and dynamic keys and doors that simultaneously move and whose logical transitions are enabled only within prescribed temporal windows.

    \item \textbf{Numerical experiments.} We demonstrate that the proposed framework successfully plans optimal trajectories satisfying STL specifications through dynamic environments with moving and time-gated task regions. Our experiment results also exhibits computational scaling governed by the complexity of the spacetime partition and the size of the resulting augmented GCS.
\end{enumerate}



\section{Problem Formulation}
\label{sec:problem}

\begin{table*}[!t]
    \centering
    \caption{The full specification library for the static case in \cite{you2026augmented} extends to dynamic environments in spacetime.}
    \begin{tabular}{ll}
        \multicolumn{1}{c}{Variation} & \multicolumn{1}{c}{Formula} \\ \hline
        \rule{0pt}{3ex}
        1. Mandatory key collection & $\varphi_{\mathcal{U}} = \bigwedge_{i=1}^{\nK} (\lnot D_i(t) \, \mathcal{U} \, K_i(t)) \land \mathcal{F}(\mathcal T)$ \\[3pt]
        2. Pure wayset collection (TSP) & $\varphi_{\mathrm{TSP}} = \bigwedge_{i=1}^{\nK} \mathcal{F}(K_i(t)) \land \mathcal{F}(\mathcal T)$ \\[3pt]
        3. Ordered key collection & $\varphi_{\mathrm{ord}} = (\lnot D_1(t) \, \mathcal{U} \, K_1(t)) \land \bigwedge_{i=2}^{\nK} \bigl(\lnot D_i(t) \, \mathcal{U} \, (K_{i-1}(t) \land \mathcal{F} K_i(t))\bigr) \land \mathcal{F}(\mathcal T)$ \\[3pt]
        4. Disjunctive keys & $\varphi_{\mathrm{disj}} = \bigwedge_{i=1}^{n_D} \bigl((K_i^A(t) \lor K_i^B(t)) \, \mathcal{R} \, \lnot D_i(t)\bigr) \land \mathcal{F}(\mathcal T)$ \\[3pt]
        5. Multi-key doors & $\varphi_{\mathrm{multi}} = \bigwedge_{i=1}^{n_D} \Bigl(\bigwedge_{j:\, K_j(t) \in \mathcal{K}_i(t)} K_j(t)\Bigr) \mathcal{R}\, \lnot D_i(t) \land \mathcal{F}(\mathcal T)$ \\[3pt]
        6. Multi-door keys & $\varphi_{\mathrm{tool}} = \bigwedge_{j=1}^{n_D} \bigl(K_{\sigma(j)}(t) \, \mathcal{R} \, \lnot D_j(t)\bigr) \land \mathcal{F}(\mathcal T)$ \\[3pt]
        7. Timed constraints & $\varphi_{\mathrm{time}} = \bigwedge_{i=1}^{\nK} \bigl(K_i(t) \, \mathcal{R}_{[a_i, b_i]} \, \lnot D_i(t)\bigr) \land \bigwedge_{i=1}^{\nK} \mathcal{F}_{[c_i, d_i]}(K_i(t)) \land \mathcal{F}(\mathcal T)$ \\[3pt]
        8. Conditional (if-then) & $\varphi_{\mathrm{cond}} = \bigwedge_{i=1}^{\nK}\Bigl(A_i \Rightarrow (K_i(t) \,\mathcal{R}\, \neg D_i(t)\Bigr) \wedge \mathcal{F}(\mathcal{T})$ \\ \hline
    \end{tabular}
    \vspace{-8pt}
    \label{tab:library}
\end{table*}

Consider a robot operating in a bounded environment $\X \subset \Real^{d}$ over a mission horizon $[0, T]$. The environment contains \emph{time-varying} obstacles $\mathcal{O} = \{O_i(t)\}_{i=1}^{\nO}$, time-varying door regions $\mathcal{D} = \{D_i(t)\}_{i=1}^{\nK}$, and time-varying key regions $\mathcal{K} = \{K_i(t)\}_{i=1}^{\nK}$. We assume a bijective key-door pairing: key $K_i$ unlocks door $D_i$ for $i = 1, \ldots, \nK$. The instantaneous obstacle-free space at time $t$ is $C(t) = X \setminus \bigcup_i O_i(t)$, and the instantaneous traversable space is $C(t) \cup K(t) \cup D(t)$, where $K(t) = \bigcup_i K_i(t)$ and $D(t) = \bigcup_i D_i(t)$. All sets are assumed to be polytopes at each instant with given half-space representations:
\begin{align*}
  X &= \{x \in \Real^d \mid H_{X} x \leq g_{X}\}, \\
  O_i(t) &= \{x \in \Real^d \mid H_{O_i}\, x \leq g_{O_i}(t)\}, \\
  D_i(t) &= \{x \in \Real^d \mid H_{D_i}\, x \leq g_{D_i}(t)\}, \\
  K_i(t) &= \{x \in \Real^d \mid H_{K_i}\, x \leq g_{K_i}(t)\}.
\end{align*}

\begin{assumption}[Piecewise-linear set motion]
We assume obstacle, door, and key sets to undergo continuous piecewise-linear motion, with each offset vector on the right-hand side of each half-space representation being piecewise-linear with up to $L$ breakpoints over the mission horizon $[0,T]$:
\[
0 = t_0 < t_1 < \cdots < t_L < t_{L+1} = T.
\]
and changing linearly within each interval $t \in [t_j,t_{j+1}]$, $j=0,\ldots,L$.
\end{assumption}

We discuss how piecewise-linear motion facilitates exact convex partitioning in spacetime in Section \ref{sec:geometry}.

\begin{remark}[Types of time-varying sets]
The formulation encompasses three qualitatively distinct dynamic behaviors. (i)~\emph{Moving obstacles}: other robots or pedestrians translate through $\X$; the offsets $g_{O_i}(t)$ encode centroid velocity. (ii)~\emph{Expanding or contracting obstacles}: a flood zone, fire perimeter, or restricted exclusion zone grows or shrinks; $g_{O_i}(t)$ encodes the expansion or contraction rate of each facet. (iii)~\emph{Moving keys and doors}: e.g., a mobile docking station or moving platform whose position must be intercepted at the right moment. 
\end{remark}

\subsection{Precedence Constraints via Signal Temporal Logic}

We use a fragment of untimed signal temporal logic (\STL) \cite{maler2004monitoring} to encode key-door precedence specifications over the time-varying environment. The base specification uses the release operator
\begin{equation}
  \varphi_{\mathrm{base}} =
    \bigwedge_{i=1}^{\nK} \bigl(K_i(t)\, \mathcal{R}\, \neg D_i(t)\bigr)
    \wedge\, \mathcal{F}(\mathcal{T}),
  \label{eq:phi_base}
\end{equation}
which requires $\neg D_i(t)$ to hold until $K_i(t)$ is visited (intercepting $K_i(t)$ releases the dynamic constraint on $D_i(t)$), while making key collection optional unless geometrically required. In the dynamic setting, the atomic propositions $K_i$ and $D_i$ are satisfied when the robot's spacetime position $(q(t), t)$ lies in the spacetime lift $\mathcal{K}_i = \{(x,t): x \in K_i(t),\, t \in [0,T]\}$ or $\mathcal{D}_i$, respectively.  The term $\mathcal{F}(\mathcal{T})$ requires the trajectory to eventually reach a target set $\mathcal T$, guaranteed by the terminal condition below.

\textbf{Time-gated keys.} A time-gated key $K_i^{[a_i, b_i]}$ is collectible only during the window $[a_i, b_i] \subseteq [0,T]$: the robot must be at $K_i(s)$ for some $t \in [a_i, b_i]$. Its spacetime lift is clipped to $\tilde{K}_i \cap (\Real^d \times [a_i, b_i])$. This captures, for example, a communication relay available only during a satellite pass, a biological sample that degrades outside a freshness window, or an authentication token valid only during a shift.

\textbf{Time-gated doors.} Similarly, a time-gated door $D_i^{[c_i, d_i]}$ is passable only during prescribed intervals; outside those intervals it acts as an ordinary obstacle. Its spacetime lift is clipped to the union of passable time windows. This captures, e.g.,  an airlock on a timed cycle, a tide-gated channel, or a moving platform that provides a crossing only when aligned.

The full specification library of STL variations established in\cite{you2026augmented} for static environments shown in Table~\ref{tab:library}, carries over to the dynamic setting by replacing spatial regions with their spacetime lifts, as detailed in Section \ref{sec:augmented}.

\enlargethispage{2\baselineskip}

\subsection{Optimization Problem}

Our goal is to find a trajectory $q$ over the horizon $[0, T]$ that solves
\begin{equation}
  \begin{aligned}
    \minimize_{q,\,T} \quad
      & \alpha \int_0^T \norm{\dot{q}(t)}_2\, dt + \beta T \\
    \text{subject to} \quad
      & q(t) \in C(t) \cup K(t) \cup D(t),
        \quad \forall\, t \in [0,T], \\
      & q(t) \models \varphi_{\mathrm{base}}(K(t), D(t)), \\
      & \dot{q}(t) \in \mathcal{Q},
        \quad \forall\, t \in [0, T], \\
      & \norm{\dot{q}(t)} \leq v_{\max},
        \quad \forall\, t \in [0,T], \\
      & q(0) = q_0,\quad q(T) = q_T, \\
      & \dot{q}(0) = \dot{q}_0,\quad \dot{q}(T) = \dot{q}_T.
  \end{aligned}
  \label{eq:problem2}
\end{equation}
The objective is a weighted sum of trajectory length and horizon duration with weights $\alpha, \beta \geq 0$. The first constraint requires the robot to remain in the instantaneous traversable space at all times. The second constraint enforces the \STL\ precedence specification over the spacetime regions $\K = \{K_i(t)\}$ and $\D = \{D_i(t)\}$. The third constraint requires the velocity to lie in the convex set $\mathcal{Q}$ at all times. The fourth constraint bounds the physical speed by $v_{\max}$; this is automatically satisfiable in the static case but becomes a genuine coupling constraint between space and time in the dynamic case. The remaining constraints specify initial and terminal conditions for position and velocity.

\begin{remark}[Reduction to static case]
When $O_i(t) \equiv O_i$, $K_i(t) \equiv K_i$, and $D_i(t) \equiv D_i$ for all $t$, problem~\eqref{eq:problem2} reduces to the static planning problem of\cite{you2026augmented}. The spacetime lifts become axis-aligned prisms $O_i \times [0,T]$, the speed constraint~$\norm{\dot{q}} \leq v_{\max}$ is inactive for sufficiently large $v_{\max}$, and the constructions in Sections \ref{sec:geometry}--\ref{sec:augmented} specialize to their static counterparts.
\end{remark}

\begin{remark}[Infinite-dimensional relaxation and B\'{e}zier parameterization]
Problem~\eqref{eq:problem2} is infinite-dimensional. As in the static case\cite{you2026augmented}, candidate trajectories are parameterized as piecewise \emph{spacetime} B\'{e}zier curves: the spatial component $q(s)$ and the temporal component $t(s)$ are both B\'{e}zier curves in the curve parameter $s \in [0,1]$, jointly forming a curve in $\Real^{d+1}$. This parameterization makes the problem finite-dimensional and admits convex constraints on higher-order derivatives. The speed constraint $\norm{\dot{q}(t)} \leq v_{\max}$ and the causality requirement $\dot{t}(s) \geq 0$ both reduce to convex (second-order cone and linear, respectively) constraints on the control points, as shown in Section\ref{sec:gcs}.
\end{remark}

\section{Spacetime Geometry}
\label{sec:geometry}
In this section, we show how piecewise-linear motion enables an exact convex partitioning in spacetime, so that the \GCS machinery can be applied directly to this spacetime partition. Under Assumption 1, consider for example obstacle $O_i(t)$ over a time interval between break points $t\in[t_{j},t_{k}]$ with $g_{O_i(t)} := g_i + H_{O_i} v_i (t - t_j) + u_i(t - t_j)$, where $v_i\in\mathbb{R}^d$ denote its translational velocity and $u_i\in\mathbb{R}^{m_i}$ the rate of change of its $m_i$ half-space offsets.
The dynamic obstacle region over interval $t\in[t_{j},t_{k}]$ is then
\[
O_i(t)
=
\left\{
x\in\mathbb{R}^d
\;\middle|\;
H_{O_i}\bigl(x-v_i(t-t_{j})\bigr)
\le
g_i+u_i(t-t_{j})
\right\}.
\]

Introducing the spacetime variable
\[
z =
\begin{bmatrix}
x\\t
\end{bmatrix}
\in\mathbb{R}^{d+1},
\]
and defining $A_i=H_{O_i}v_i+u_i$, the swept region over $[t_{j},t_{k}]$ has the static spacetime lift
\[
\mathcal{O}_i
=
\left\{
z\in\mathbb{R}^{d+1}
\;\middle|\;
H^{ST}_iz\le g^{ST}_i
\right\},
\]
where
\[
H^{ST}_i =
\begin{bmatrix}
H_{O_i} & -A_i\\
0_{1\times d} & -1\\
0_{1\times d} & 1
\end{bmatrix},
\qquad
g^{ST}
=
\begin{bmatrix}
g_i-A_it_{j}\\
-t_{j}\\
t_{k}
\end{bmatrix}.
\]
Thus, a convex polytope undergoing constant translation and affine motion of its half-space offsets lifts exactly to a convex polytope in spacetime. This model includes pure translation ($u=0$), isotropic or anisotropic dilation, and their combinations, as illustrated in Fig.~\ref{fig:dynamic_regions}.
\begin{figure}[t]
    \centering
    \begin{subfigure}[t]{0.4\linewidth}
        \centering
        \includegraphics[width=\linewidth]{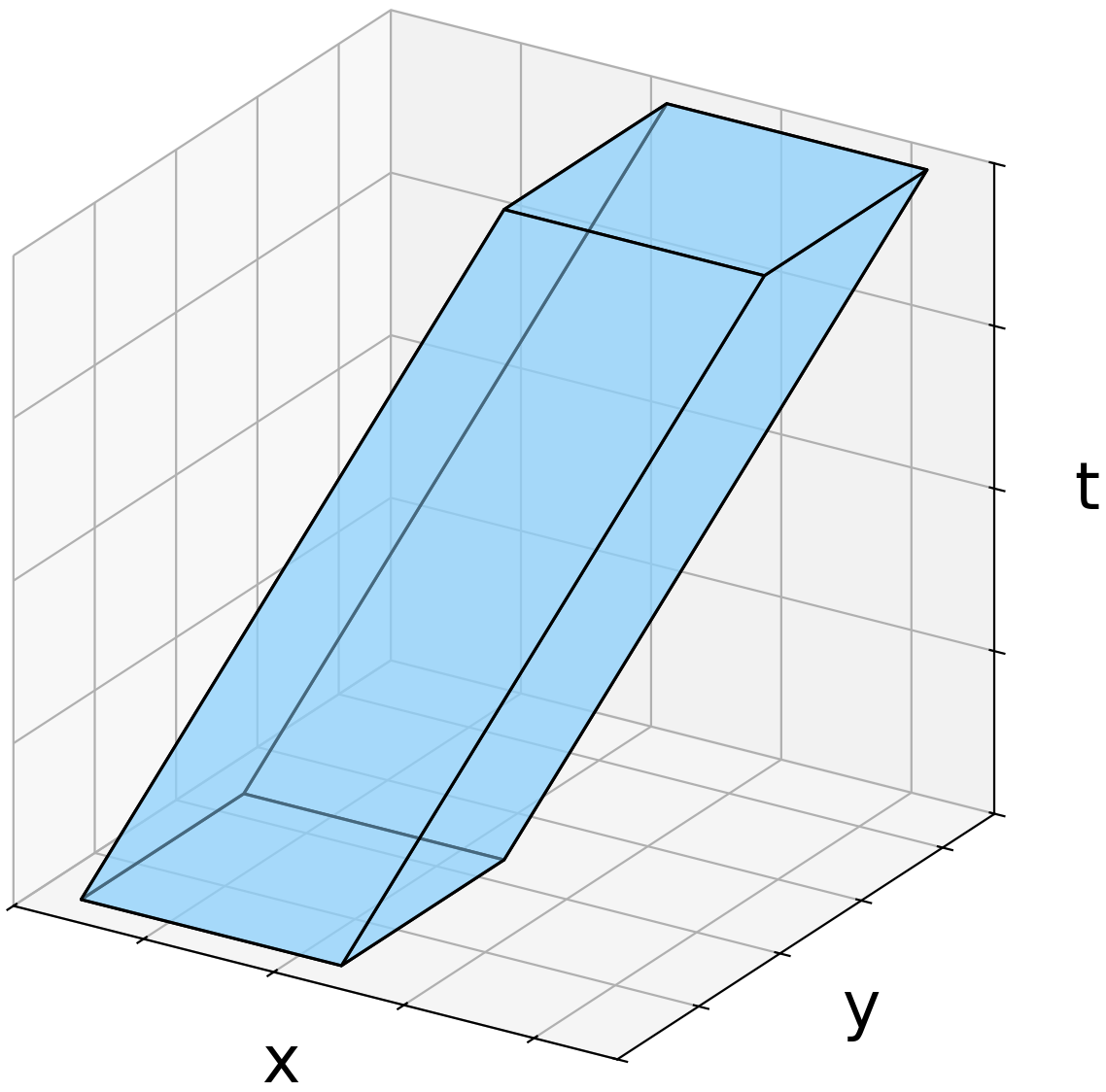}
        \caption{Pure translation}
        \label{fig:pure_translation}
    \end{subfigure}
    \hspace{0.08\linewidth}
    \begin{subfigure}[t]{0.4\linewidth}
        \centering
        \includegraphics[width=\linewidth]{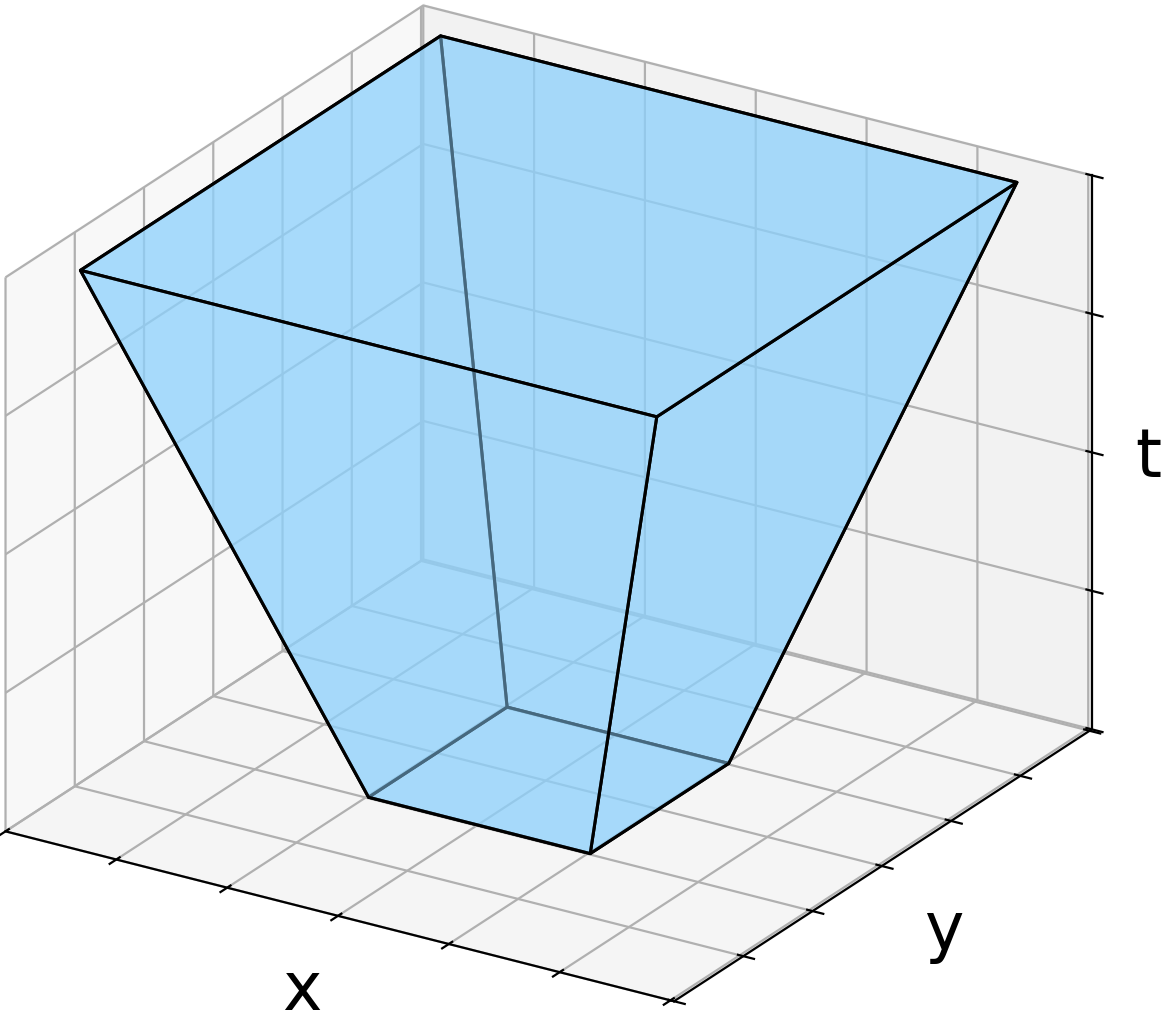}
        \caption{Isotropic dilation}
        \label{fig:isotropic_dilation}
    \end{subfigure}

    \begin{subfigure}[t]{0.4\linewidth}
        \centering
        \includegraphics[width=\linewidth]{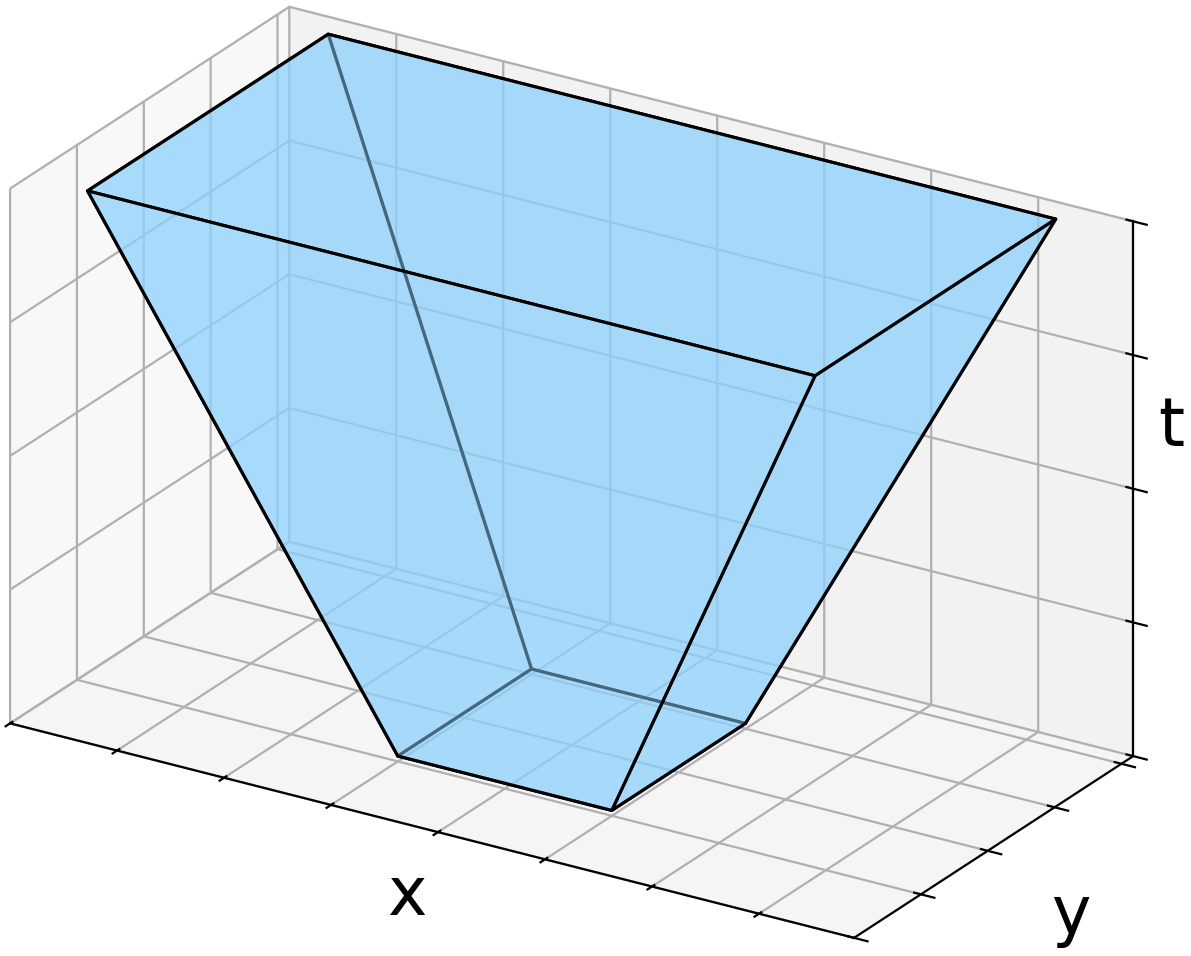}
        \caption{Anisotropic dilation}
        \label{fig:anisotropic_dilation}
    \end{subfigure}
    \hspace{0.08\linewidth}
    \begin{subfigure}[t]{0.4\linewidth}
        \centering
        \includegraphics[width=\linewidth]{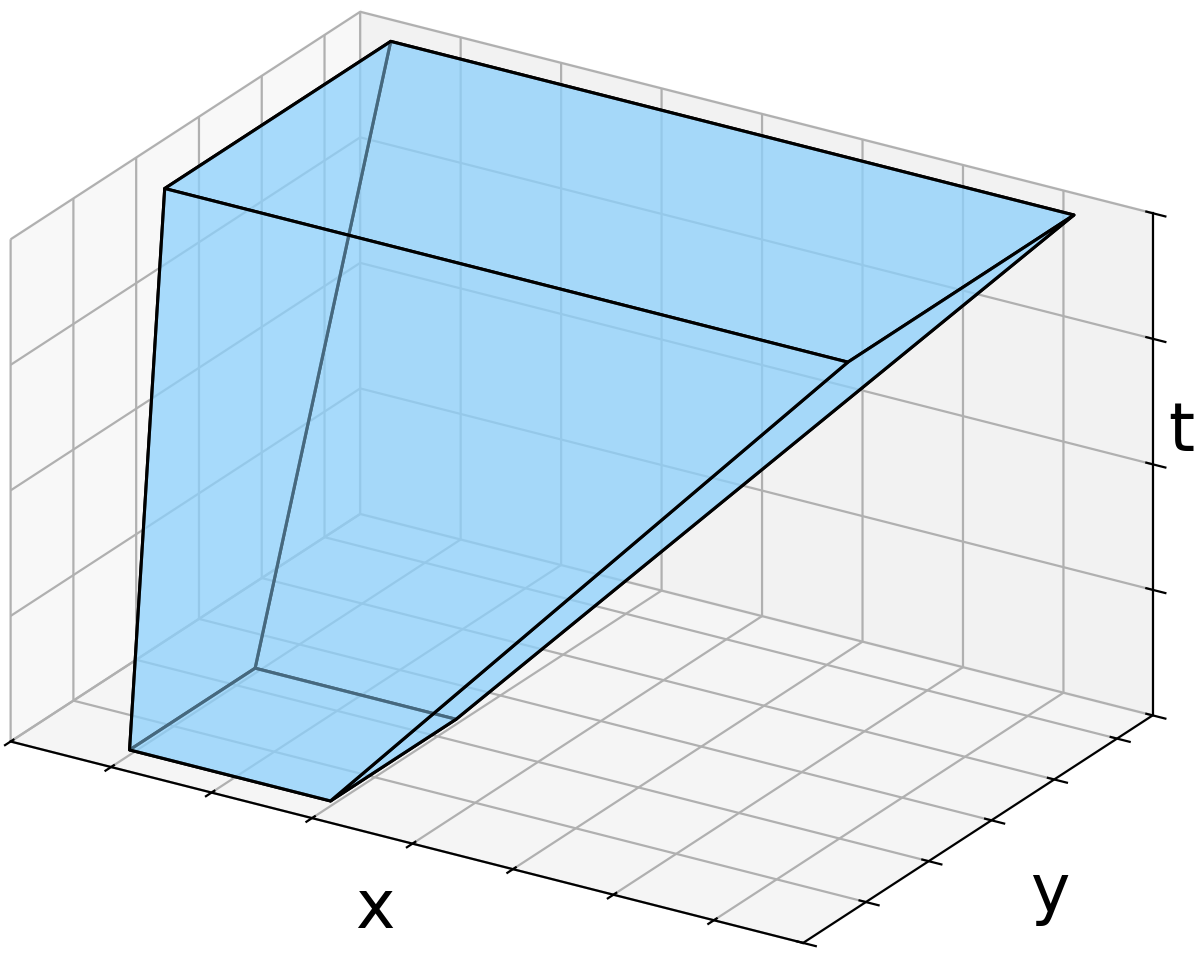}
        \caption{Combined motion}
        \label{fig:combined_motion}
    \end{subfigure}

    \caption{Spacetime lifts of a square undergoing different types of piecewise-linear motion.}
    \label{fig:dynamic_regions}
    \vspace{-10pt}
\end{figure}

The construction can be applied to each interval between breakpoints and for each time-varying obstacle, key, and door region. We denote the spacetime lifts of all obstacle, key, and door regions by $\mathcal{O} =  \{O_i(t)\}_{i=1}^{\nO}$, $\mathcal{K} =  \{ K_i(t)\}_{i=1}^{\nK}$, and $\mathcal{D} =  \{ D_i(t)\}_{i=1}^{\nK}$, respectively. This produces a finite collection of convex spacetime polytopes. Hence, piecewise-linear motion of obstacles, keys, and doors produces a spacetime free space region that admits an exact finite polytopic representation, discussed next.

\subsection*{Exact Convex Partition in Spacetime}
\label{sec:spacetime_partition}

The piecewise-linear motion and resulting polytopic spacetime representation allows the exact convex partitioning algorithm developed for the static setting (Algorithm 1 in~\cite{you2026augmented}, based on reverse search \cite{avis1996reverse}), to be applied directly in $\mathbb{R}^{d+1}$. The boundaries of the lifted environment, obstacles, and doors define a hyperplane arrangement whose free cells form an exact convex partition of the spacetime free space. As in the static case, adjacent cells may subsequently be merged while preserving convexity, and the resulting cells and their adjacencies define the spacetime GCS. Key regions, which do not restrict traversability, are incorporated afterward by labeling the spacetime cells that intersect them. Therefore, no modification to the partitioning algorithm is required; the transition from a static to a dynamic environment changes only the dimension and geometry of its input polytopes.

Fig.~\ref{fig:partition_example} illustrates this construction for a two-dimensional environment lifted to three-dimensional spacetime.

\begin{figure}[t]
    \centering
    \begin{subfigure}[t]{0.48\linewidth}
        \centering
        \includegraphics[width=\linewidth]{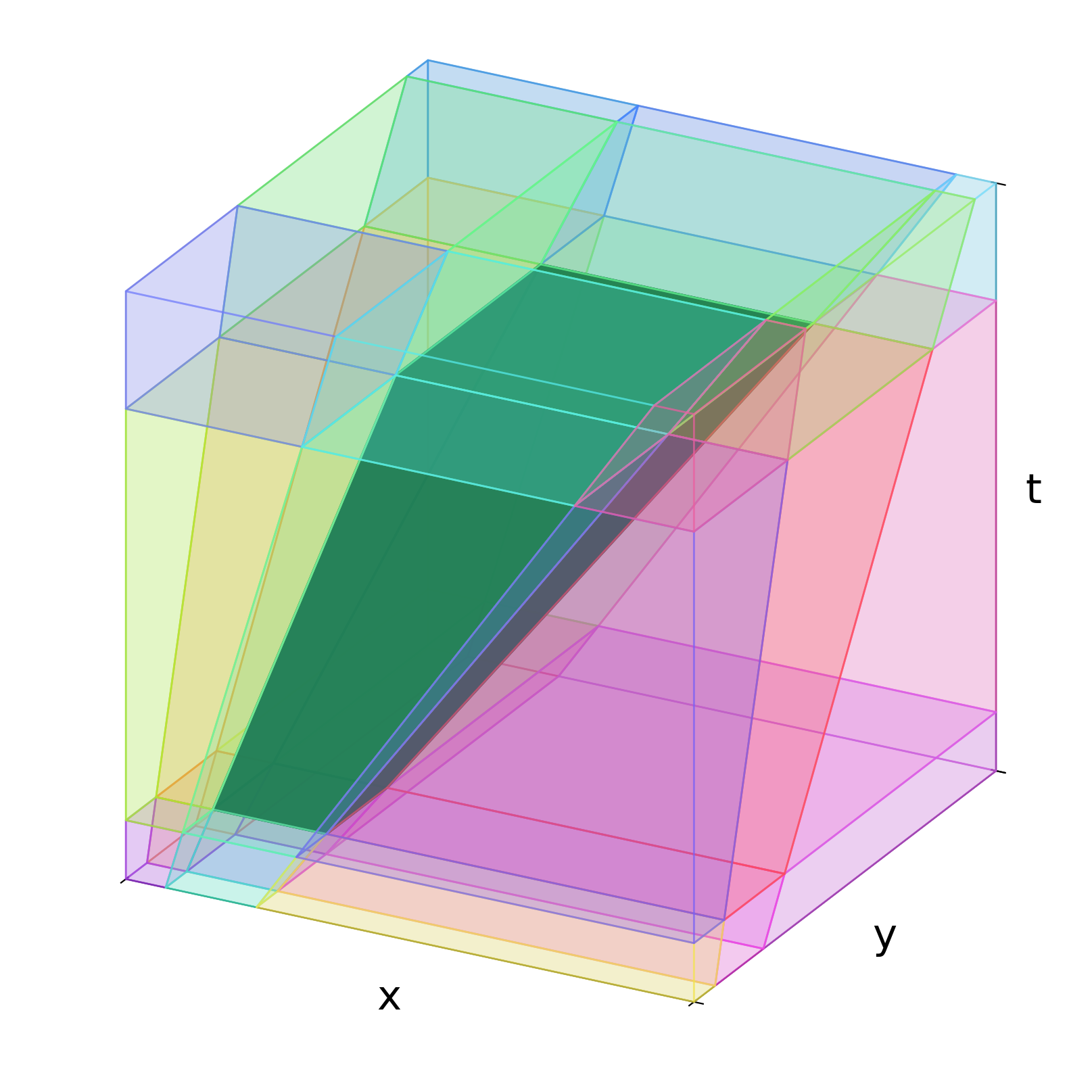}
        \caption{Before merging.}
        \label{fig:partition_unmerged}
    \end{subfigure}
    \begin{subfigure}[t]{0.48\linewidth}
        \centering
        \includegraphics[width=\linewidth]{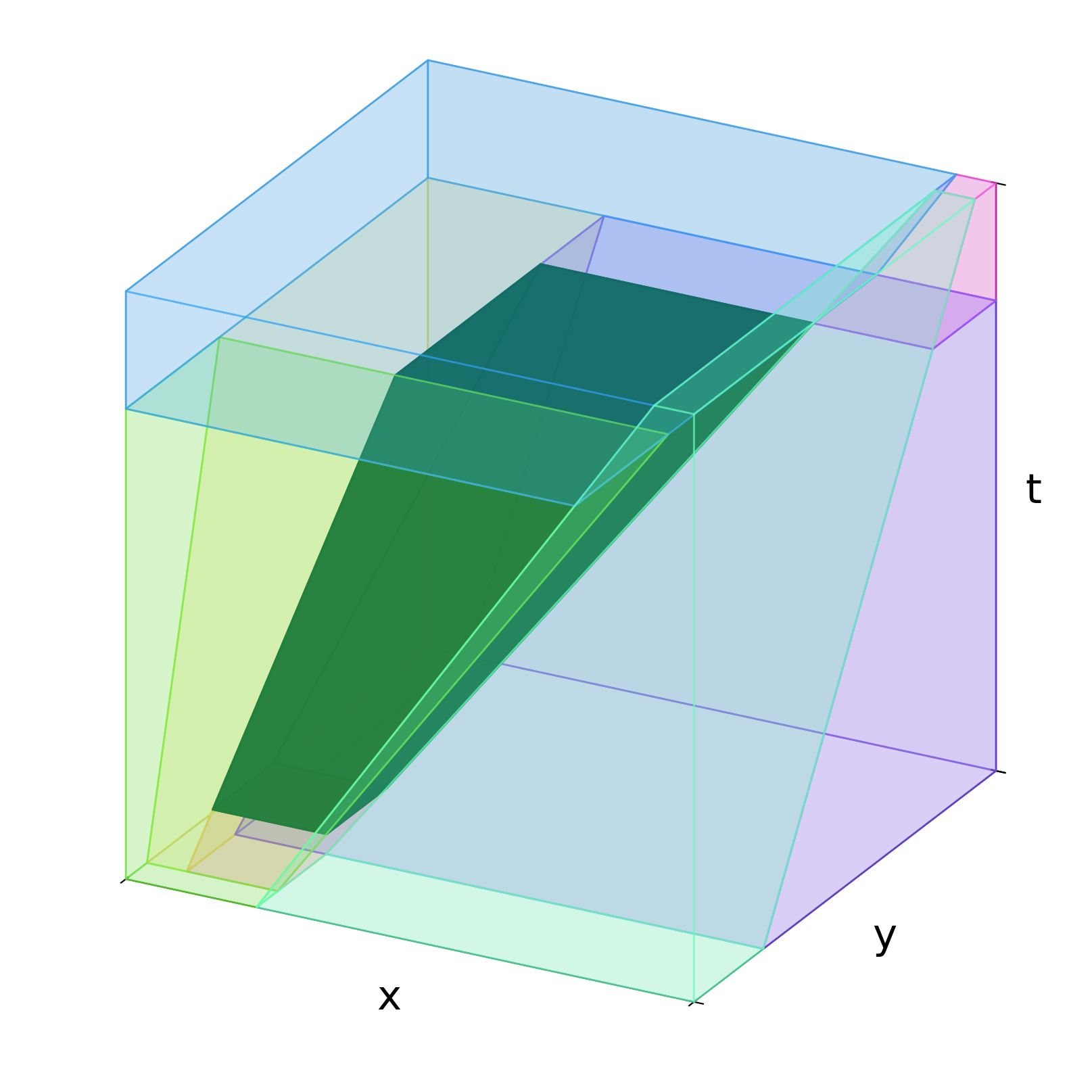}
        \caption{After merging.}
        \label{fig:partition_merged}
    \end{subfigure}
    \caption{Convex partition of a spacetime environment containing a translating and expanding obstacle.}
    \label{fig:partition_example}
    \vspace{-10pt}
\end{figure}

\begin{remark}
    The fixed orientation (constant facet normal matrices) excludes continuous rotation, even if the rotational motion is linear in time. Under rotation, the facet normals become time-varying and the resulting spacetime constraints are generally nonlinear in $(x,t)$; consequently, the swept region need not be convex or polytopic. Rotational motion is therefore not considered in this work.
\end{remark}

\section{Spacetime Augmented GCS}

\subsection{Spacetime Trajectory Parameterization}
\label{sec:gcs}

Given the spacetime partition constructed in Section~\ref{sec:geometry}, we form a spacetime GCS\ $\mathcal{G}=(\mathcal{V},\mathcal{E})$, where each vertex $v\in\mathcal{V}$ corresponds to a convex spacetime cell $\mathcal{P}_v\subset\mathbb{R}^{d+1}$ and edges connect adjacent cells. We adopt the spacetime trajectory parameterization used in prior work~\cite{de2026stl}, with each trajectory segment represented by a degree-$n$ B\'ezier curve
\[
\gamma(s)
=
\begin{bmatrix}
r(s)\\
h(s)
\end{bmatrix}
=
\sum_{i=0}^{n} B_{i,n}(s)
\begin{bmatrix}
p_i\\
t_i
\end{bmatrix},
\qquad s\in[0,1],
\]
where $B_{i,n}(s)=\binom{n}{i}s^i(1-s)^{n-i}$ is the $i$th Bernstein basis polynomial of degree $n$, and $(p_i,t_i)\in\mathbb{R}^{d+1}$ are the spacetime control points, with $p_i\in\mathbb{R}^d$ and $t_i\in\mathbb{R}$ denoting their spatial and temporal coordinates, respectively.

The spacetime B\'ezier curve can therefore be decomposed into a spatial
curve and a temporal curve,
\[
r(s)=\sum_{i=0}^{n}B_{i,n}(s)p_i,
\qquad
h(s)=\sum_{i=0}^{n}B_{i,n}(s)t_i.
\]
The spatial curve $r(s)$ determines the geometric shape of the trajectory in $\mathbb{R}^d$, while the temporal curve $h(s)$ determines its time scaling, mapping the curve parameter $s$ to physical time.

To preserve causality and ensure that this mapping is invertible, $h(s)$
must be strictly increasing. Its derivative with respect to $s$ is
\[
h'(s) =
n\sum_{i=0}^{n-1}B_{i,n-1}(s)(t_{i+1}-t_i).
\]
We enforce strict monotonicity by imposing
\[
t_{i+1}-t_i\geq\epsilon,
\qquad i=0,\ldots,n-1,
\]
for some $\epsilon>0$. Since the Bernstein basis functions are nonnegative and form a partition of unity, this guarantees $h'(s)\geq n\epsilon>0$ for all $s\in[0,1]$, and hence $h(s)$ is strictly increasing. The corresponding physical trajectory is
$q(t)=r(h^{-1}(t))$, with
\[
\dot q(t)=\frac{r'(s)}{h'(s)},
\qquad s=h^{-1}(t).
\]

For a segment assigned to a spacetime cell
$\mathcal{P}_v$, all spacetime control points are constrained to lie
in $\mathcal{P}_v$. By the convex-hull property of B\'ezier curves,
this guarantees that the entire segment remains within the corresponding
collision-free spacetime region. Velocity bounds, boundary conditions,
continuity constraints, and trajectory costs are imposed on the spatial and
temporal components using the standard convex formulations
of~\cite{marcucci2023motion,de2026stl}. Thus, once the dynamic environment has
been converted into a convex spacetime partition, no new trajectory
parameterization is required. The output is a labeled graph of spacetime convex sets $\mathcal{G}(\mathcal{V}, \mathcal{E})$ with vertex partition $\mathcal{V} = \{\mathcal{V}_\mathcal{C}, \mathcal{V}_\mathcal{K}, \mathcal{V}_\mathcal{D} \}$ corresponding to free space, key, and door regions. The polytopes corresponding to $\mathcal{V}_\mathcal{C}$ form an exact convex partition of the spacetime free space, and the edge set encodes connectivity among free-space, key, and door regions.

\begin{remark}[Graph connectivity and dynamic feasibility]
\label{rem:dynamic_feasibility}
An important distinction from static GCS is that graph connectivity alone does not guarantee dynamic feasibility in the spacetime GCS. A discrete path may traverse a spacetime cell whose spatial extent, relative to its temporal extent, requires a velocity exceeding the prescribed bound. Consequently, a source-to-target path may exist in the spacetime GCS while the corresponding continuous trajectory how to optimization is infeasible.
\end{remark}

\subsection{Spacetime Augmented GCS}
\label{sec:augmented}

The spacetime GCS introduced in Section~\ref{sec:gcs} encodes collision-free motion in a dynamic environment by representing every obstacle-free spacetime cell as a graph vertex. To additionally enforce temporal logic precedence specifications, we construct an \emph{augmented spacetime GCS} using exactly the same layered construction developed for the static environment. The only difference is that each subgraph is now a copy of the \emph{spacetime} GCS rather than the spatial GCS.

As in the static formulation, every subgraph corresponds to one reachable subset of collected keys. Let $S \subseteq \mathcal{K}$ denote the set of keys that have already been collected. The corresponding subgraph is denoted by
\[
\mathcal{G}_S=
(\mathcal{V}_S,\mathcal{E}_S),
\]
where $\mathcal{V}_S=\mathcal{V}$, that is, every subgraph is an identical copy of the base spacetime GCS. The difference between subgraphs lies only in the enabled door transitions.

Consequently, the combinatorial structure of the augmented graph is unchanged from the static case; only the geometry represented by each graph vertex has changed from a spatial convex region to a spacetime convex polytope.

\paragraph{Base Layer ($S=\emptyset$)}

The base layer corresponds to the state in which no keys have been collected.
The vertex set is identical to the base spacetime GCS, so $\mathcal{V}_{\emptyset}= \mathcal{V}$,
while all edges incident to door regions are removed
\[
\mathcal{E}_{\emptyset}=
\mathcal{E}
\setminus
\left\{
(u,v)\in\mathcal{E}
\mid
u\in\mathcal{V}_{\mathcal{D}}
\text{ or }
v\in\mathcal{V}_{\mathcal{D}}
\right\}.
\]

Therefore, every door remains inaccessible before its corresponding key has been collected, even if the door itself becomes geometrically reachable in spacetime. A start node is instantiated and connected to a free space vertex whose corresponding set contains the initial condition.

\paragraph{Intermediate Layers}
For every reachable key subset $S\subseteq\mathcal K$, a copy of the spacetime GCS is created. Its vertex set is again $\mathcal{V}_S=\mathcal{V}$ while its edge set is
\[
\mathcal{E}_S=
\mathcal{E}_{\emptyset}
\cup
\left\{
(u,v)\in\mathcal{E}
\mid
u\in\mathbf D(S)
\text{ or }
v\in\mathbf D(S)
\right\},
\]
where $\mathbf D(S)$ denotes the set of door regions unlocked by the keys contained in $S$.

Whenever a trajectory reaches the spacetime key region corresponding to key $\mathcal{K}_i$, a zero-cost transition is inserted from the copy of that key vertex in $\mathcal{G}_{S\setminus\{i\}}$ to the corresponding key vertex in $\mathcal{G}_{S}$. This edge represents the logical event of collecting a key and advancing the specification state. Unlike the static case, reaching a key now requires satisfying both its spatial and temporal conditions simultaneously. Since every key is represented by a spacetime polytope, the transition between subgraphs is enabled only when the
trajectory enters that spacetime region. Consequently, time-gated keys and moving keys require no modification to the augmentation algorithm; they are handled entirely through the geometry of the underlying spacetime GCS.

\paragraph{Final Layer}
The final layer corresponds to the full key set $S=\mathcal K$. 
Its edge set equals the complete spacetime GCS $\mathcal{E}_{\mathcal K}= \mathcal{E}$,
since every door has been unlocked. As in the static construction, zero-cost transitions connect every predecessor subgraph containing all but one key to
the final subgraph.
\paragraph{Target Merging}
Let the target spacetime region be represented by a vertex
\[
\mathcal{T}\in\mathcal{V}_{\mathcal{C}}.
\]
All copies of this target vertex are merged into a single terminal node of the augmented spacetime GCS.

We denote the spacetime augmented \GCS by $\hat{\mathcal{G}} = (\hat{\mathcal{V}}, \hat{\mathcal{E}})$, where the vertex and edge sets contains all corresponding copies across all subgraphs $\hat{\mathcal{V}} = \cup_{S} \mathcal{V}_S$ and $\hat{\mathcal{E}} = \cup_{S} \mathcal{E}_S$. Our main results generalize the augmented \GCS framework to the dynamic spacetime setting.

\begin{theorem}[Optimality of Spacetime Augmented \GCS Shotest Path]
\label{thm:optimality}
A shortest path from start node $s$ to the merged target node $t$ in the spacetime augmented  GCS $\hat{\mathcal{G}} = (\hat{\mathcal{V}}, \hat{\mathcal{E}})$ corresponds to an optimal solution of problem \eqref{eq:problem2}, up to a finite parameterization of the trajectory $q$ by B\'ezier curves.
\end{theorem}
\begin{proof}
Identical to the proof for the static case in \cite{you2026augmented}, with the augmented \GCS replaced by the spacetime augmented \GCS. 
\end{proof}

\begin{theorem}[Structural Bellman-Held-Karp Correspondence and Singly Exponential Size]
\label{thm:spacetime-augmentation-complexity}
The layer structure and inter-layer direction edges of the spacetime augmented \GCS implement the subset enumeration and recurrence relation of the Bellman-Held-Karp algorithm \cite{bellman1962dynamic,held1962dynamic}, generalized from given discrete transition cost to continuous time-varying geometry.
For $\nK$ keys, the spacetime augmented \GCS contains at most $2^{\nK}$ subgraphs, one for each reachable key subset. 
\end{theorem}

\begin{proof}
Identical to the proof for the static case in \cite{you2026augmented}, with the augmented \GCS replaced by the spacetime augmented \GCS. 
Since each subgraph is indexed by a subset $S\subseteq\mathcal K$, there are at most $|2^{\mathcal K}|=2^{\nK}$ subgraphs, which is achieved when all keys are reachable from the start node. 
\end{proof}

Algorithm~\ref{alg:spacetime_augmented_gcs} summarizes the overall planning
pipeline. The dynamic geometry is first converted into a static convex
decomposition of spacetime. The resulting spacetime \GCS is then augmented
to encode the precedence specification and solved jointly for its discrete
path and continuous trajectory.

\begin{algorithm}[t]
\caption{Spacetime Augmented GCS Planning}
\label{alg:spacetime_augmented_gcs}
\begin{algorithmic}[1]
\REQUIRE Environment domain $X$, dynamic obstacles $\{O_i(t)\}$,
         doors $\{D_i(t)\}$, keys $\{K_i(t)\}$, planning horizon $[0,T]$,
         and precedence specification $\varphi_{\mathrm{base}}$
\ENSURE An optimal trajectory $q^*(t)$ for Problem \eqref{eq:problem2}

\STATE Lift $X$, $\{O_i(t)\}$, $\{D_i(t)\}$, and $\{K_i(t)\}$ into
       spacetime polytopes $\mathcal{X}$, $\{\mathcal{O}_i\}$,
       $\{\mathcal{D}_i\}$, and $\{\mathcal{K}_i\}$. (Section \ref{sec:geometry})

\STATE Partition the spacetime free space into convex regions and merge
       adjacent regions to construct the base spacetime \GCS, and associate a graph of convex sets $\mathcal{G}(\mathcal{V}, \mathcal{E})$ with vertex partition $\mathcal{V} = \{\mathcal{V}_\mathcal{C}, \mathcal{V}_\mathcal{K}, \mathcal{V}_\mathcal{D} \}$ with a vertex for each free space, key, and door region, and edges representing adjacency amongst regions. (Algorithm 1 in \cite{you2026augmented})

\STATE Construct the augmented spacetime GCS  $\hat{\mathcal{G}} = (\hat{\mathcal{V}}, \hat{\mathcal{E}})$, exactly encoding the precedence specification
       $\varphi$. (Section ~\ref{sec:augmented})

\STATE Solve the resulting shortest-path problem over the augmented
       spacetime GCS to obtain $q^*(t)$. (via mixed-integer convex reformulation from \cite{marcucci2024shortest,marcucci2023motion})

\RETURN An optimal trajectory $q^*(t)$ satisfying $\varphi_{\mathrm{base}}$.
\end{algorithmic}
\end{algorithm}

\section{Numerical Experiments}
\label{sec:experiments}

In this section, we demonstrate and evaluate the augmented GCS framework for solving simple and complex dynamic environment with temporal logic. Experiments were performed on a laptop with an Apple M2 CPU and 8GB RAM. The underlying convex optimization solver was MOSEK \cite{mosek}, accessed via the GCS components from the Drake \cite{drake} Python bindings. We use the rounding implementation from \cite{marcucci2023motion} with 100 trials and 10 max paths. 

\subsection{Untimed Dynamic Key-Door Problem}
\label{subsec:2key2door}
The first problem we consider is an environment with two keys and two doors, the goal is to reach the target within 5 seconds, so the corresponding signal temporal logic is
\[
\varphi=
\bigwedge_{i=1}^{2}
\left(
K(t)_i
\;\mathcal{R}\;
\neg D(t)_i
\right)
\;\land\;
\mathcal{F}_{[0,6]}(\mathcal{T}).
\]
More specifically, in this example both key sets and one obstacle move over time. In both subplots of Fig.~\ref{fig:2key2door_solutions}, we show the initial positions of all labeled sets together with the regions they sweep over the entire planning horizon.

The two subplots illustrate the effect of different objective functions. Figure~\ref{fig:2key2door_distance} shows the optimal trajectory obtained when only the total trajectory length is minimized, corresponding to $\alpha=1$ and $\beta=0$ in Problem~\ref{eq:problem2}. As a result, the optimizer uses the entire 6-second time horizon to reach the target and intentionally moves slowly at the beginning, effectively “waiting” for the moving key sets to approach before collecting them.

In contrast, Fig.~\ref{fig:2key2door_time} shows the optimal trajectory obtained when only the total travel time is minimized, corresponding to $\alpha=0$ and $\beta=1$. Rather than waiting for the key sets to move closer, the robot sacrifices path length and follows an S-shaped trajectory to intercept the moving key sets earlier. Consequently, it completes the task in approximately 3 seconds. Left column of Table~\ref{tab:2key2door_time} shows the computation time for this problem.

\begin{figure}[t]
    \centering
    \begin{subfigure}{0.48\linewidth}
        \centering
        \includegraphics[width=\linewidth]{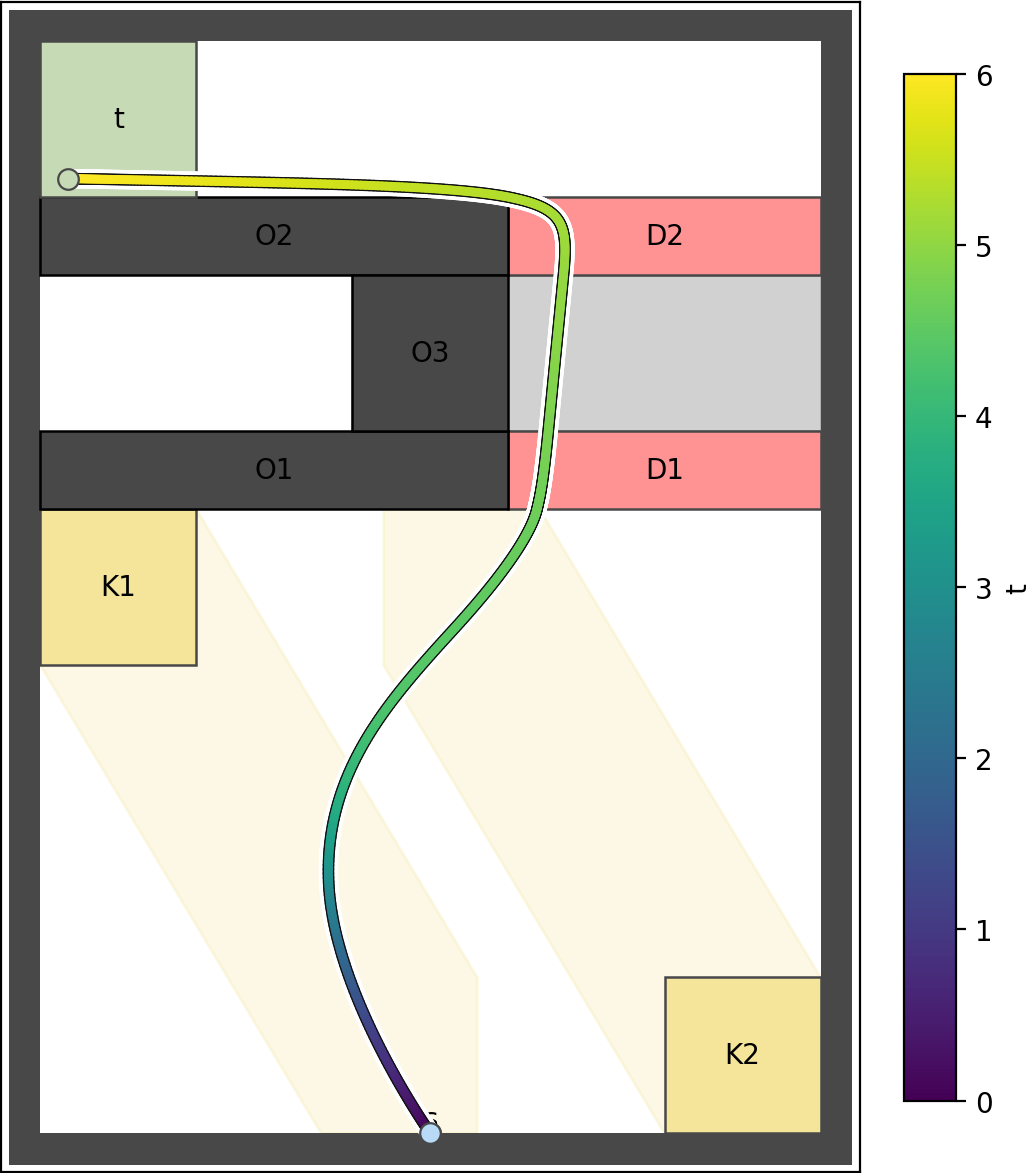}
        \caption{}
        \label{fig:2key2door_distance}
    \end{subfigure}
    \hfill
    \begin{subfigure}{0.48\linewidth}
        \centering
        \includegraphics[width=\linewidth]{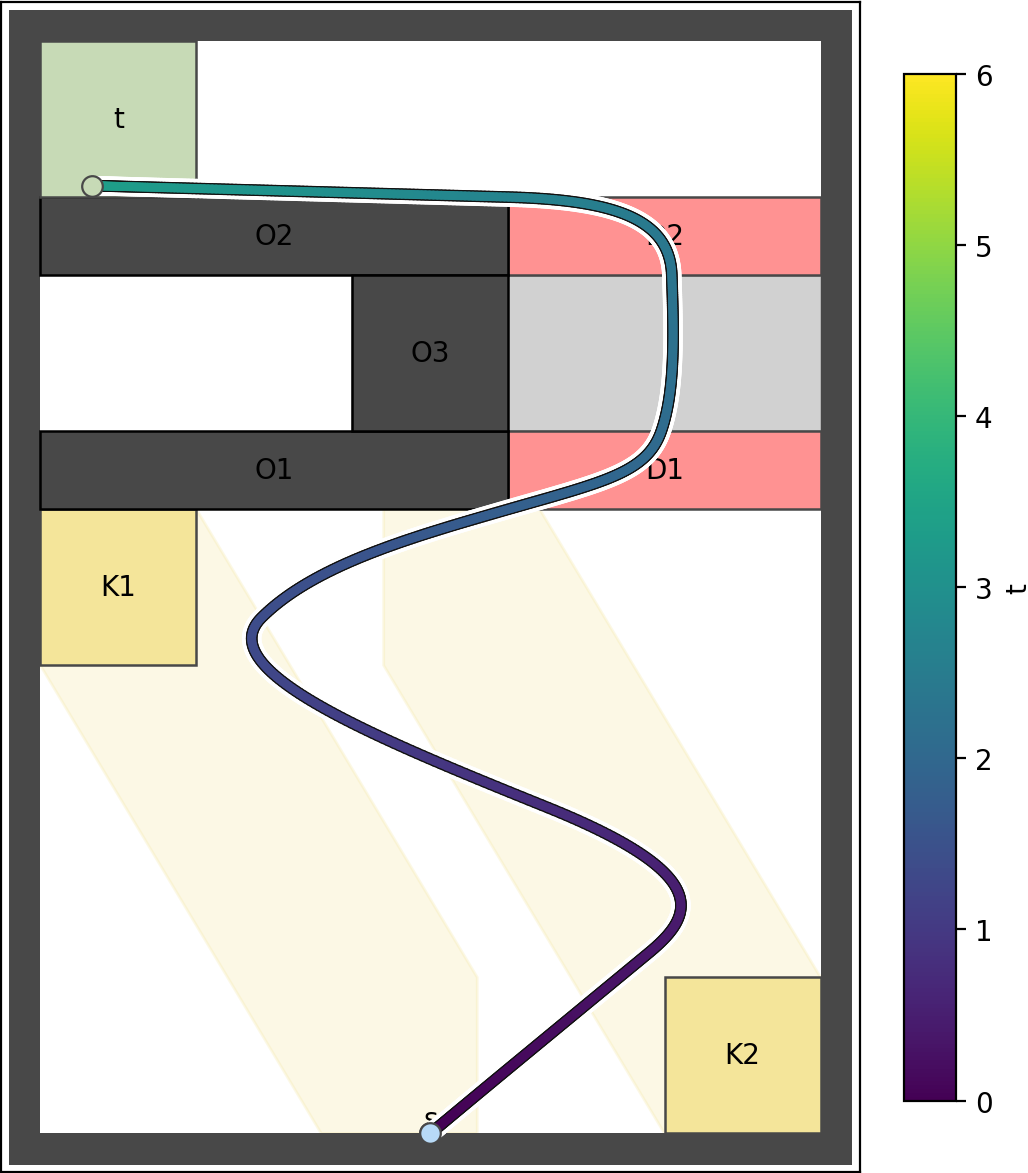}
        \caption{}
        \label{fig:2key2door_time}
    \end{subfigure}
    \caption{Optimal trajectories of a dynamic environment with 2 moving keys and 2 doors minimizing (a) total distance; (b) total travel time.}
    \label{fig:2key2door_solutions}
    \vspace{-10pt}
\end{figure}

\subsection{Time-gated Dynamic Key-Door Problem}
\label{subsec:2key2door_timed}
In this subsection, we show a modified experiment of Section~\ref{subsec:2key2door}, the environment remains the same except that $\mathcal{K}_1$ appears during time interval $t=[0,4]$ and $\mathcal{K}_1$ appears during time interval $t=[2,6]$. The \STL specification is
\[
\begin{aligned}
\varphi ={}&
\mathcal{F}_{[2,4]}(K(t)_1)
\land
(K(t)_1 \mathcal{R} \neg D(t)_1)
\\
&{}\land
\mathcal{F}_{[2,6]}(K(t)_2))
\land
(K(t)_2 \mathcal{R} \neg D(t)_2)
\\
&{}\land
\mathcal{F}_{[0,6]}(\mathcal{T}).
\end{aligned}
\]
The right column of Table~\ref{tab:2key2door_time} shows the computation time for this problem, and Fig.~\ref{fig:timed_2key2door_solutions} shows the optimal  trajectories.

\begin{figure}[t]
    \centering
    \begin{subfigure}{0.48\linewidth}
        \centering
        \includegraphics[width=\linewidth]{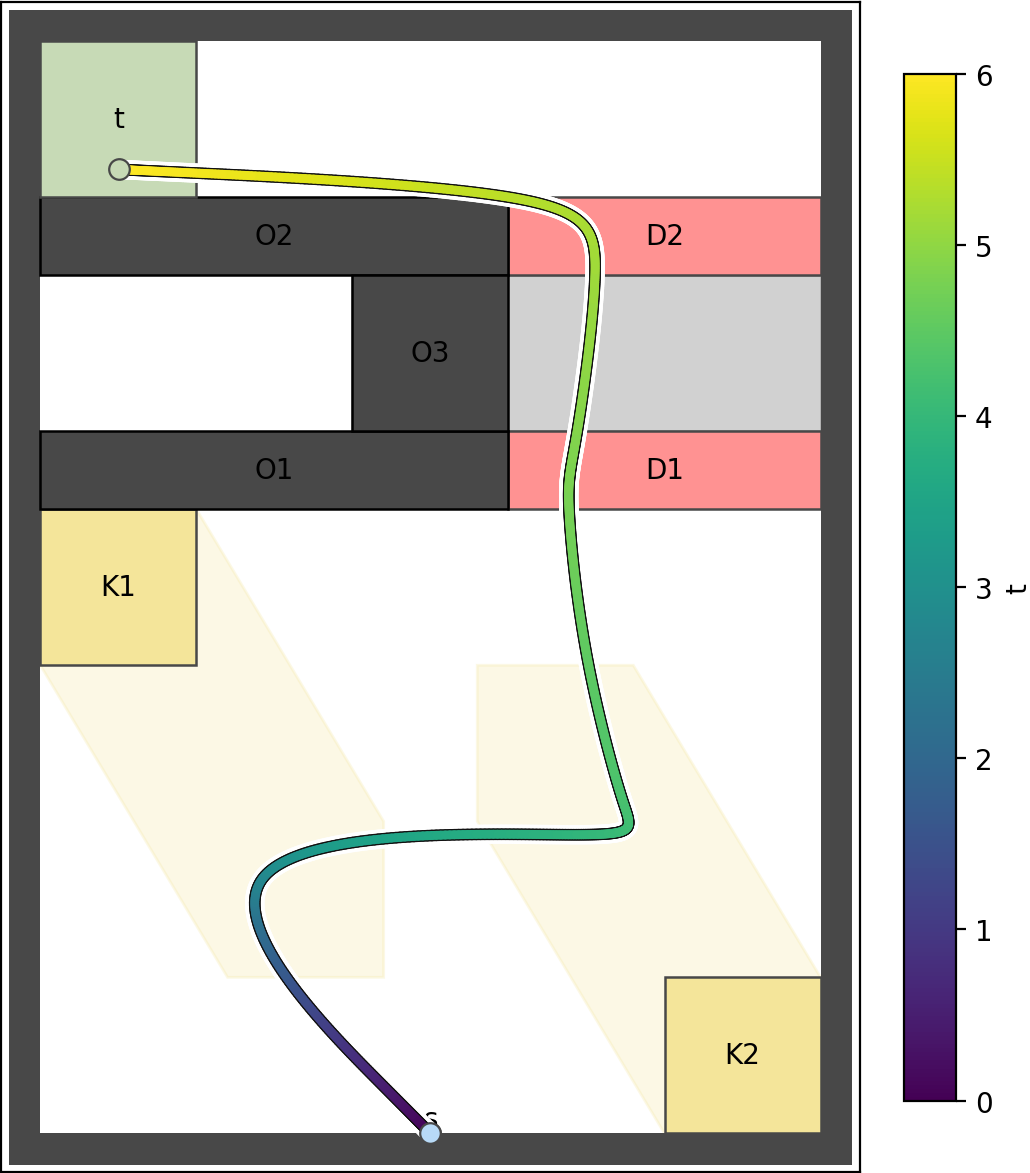}
        \caption{}
        \label{fig:timed_2key2door_distance}
    \end{subfigure}
    \hfill
    \begin{subfigure}{0.48\linewidth}
        \centering
        \includegraphics[width=\linewidth]{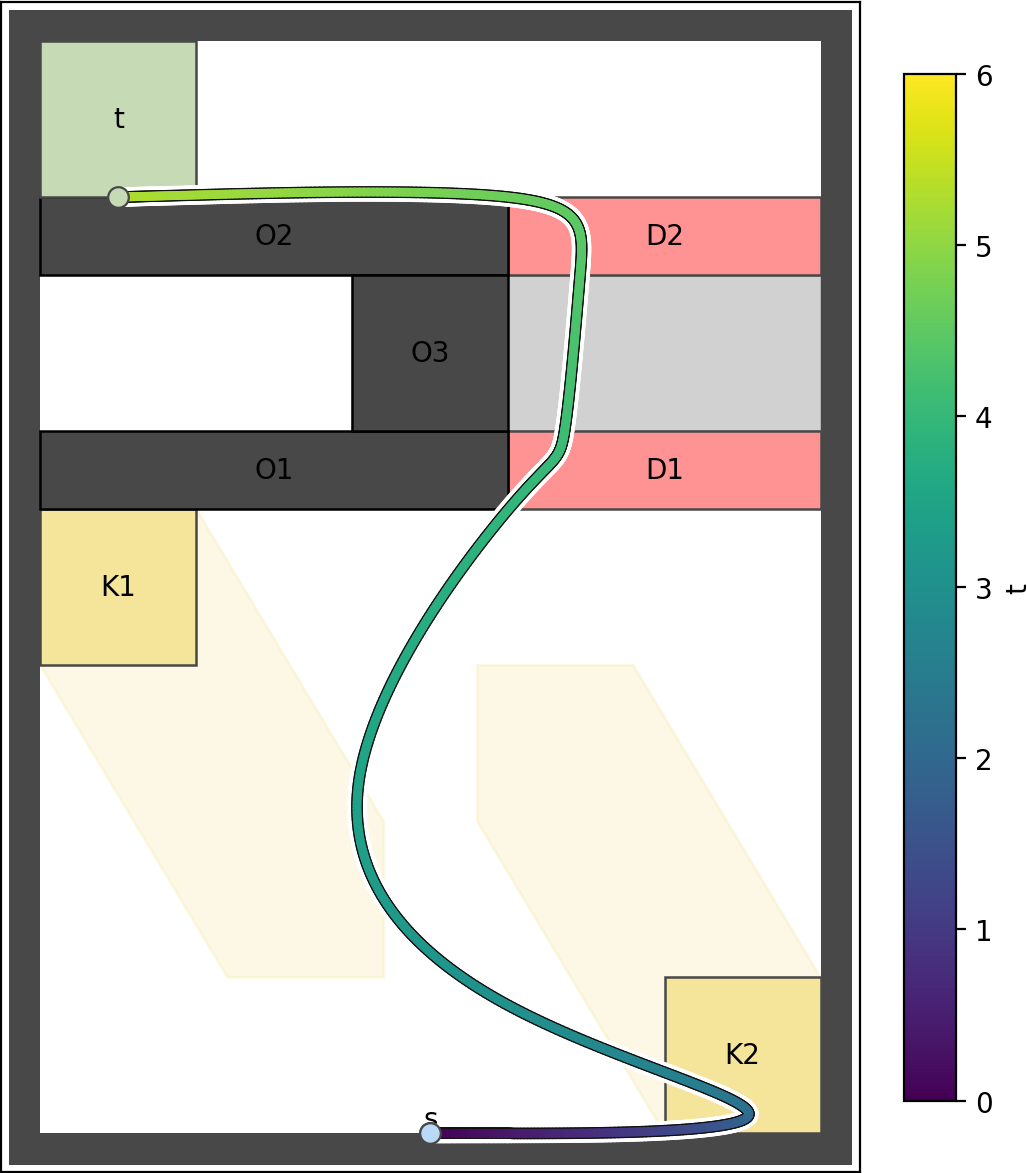}
        \caption{}
        \label{fig:timed_2key2door_time}
    \end{subfigure}
    \caption{Optimal trajectories of a dynamic environment with 2 moving and time-gated keys and 2 doors minimizing (a) total distance; (b) total travel time.}
    \label{fig:timed_2key2door_solutions}
    \vspace{-5pt}
\end{figure}

\begin{table}[t]
    \centering
    \caption{Computation time comparison for untimed and timed task regions}
    \label{tab:2key2door_time}
    \begin{tabular}{lrr}
    \hline
    & \multicolumn{2}{c}{Computation Time [s]} \\
    Stage & Untimed & Time-gated \\
    \hline
    \rule{0pt}{2ex}
    Form $\mathcal{G}(\mathcal{V}, \mathcal{E})$    & 0.250    & 0.285 \\
    Form $\hat{\mathcal G} (\hat{\mathcal V}, \hat{\mathcal E})$  & 0.0003 & 0.0004 \\
    Solve & 0.387    & 0.363 \\
    \hline
    Total           & 0.638    & 0.649 \\
    \end{tabular}
    \vspace{-10pt}
\end{table}

\begin{table*}[!t]
    \centering
    \caption{Scalability results for the augmented spacetime GCS framework.}
    \label{tab:scalability}
    \begin{tabular}{
    >{\centering\arraybackslash}m{0.03\textwidth}
    >{\centering\arraybackslash}m{0.03\textwidth}
    >{\centering\arraybackslash}m{0.055\textwidth}
    >{\centering\arraybackslash}m{0.055\textwidth}
    >{\centering\arraybackslash}m{0.055\textwidth}
    >{\centering\arraybackslash}m{0.10\textwidth}
    >{\centering\arraybackslash}m{0.12\textwidth}
    >{\centering\arraybackslash}m{0.07\textwidth}
    >{\centering\arraybackslash}m{0.07\textwidth}
    }
    \toprule
    \multirow{2}{*}{ID}
    & \multirow{2}{*}{\makecell{\# of\\Keys}}
    & \multirow{2}{*}{\makecell{\# of\\Subgraphs}}
    & \multirow{2}{*}{$|\hat{\mathcal{V}}|$}
    & \multirow{2}{*}{$|\hat{\mathcal{E}}|$}
    & \multicolumn{4}{c}{Computation Time [s]} \\
    \cmidrule{6-9}
    & & & & &
    Form $\mathcal{G}(\mathcal{V}, \mathcal{E})$
    & Form $\hat{\mathcal G} (\hat{\mathcal V}, \hat{\mathcal E})$
    & Solve
    & Total \\
    \midrule
    1 & 2 & 4  & 53  & 105  & 0.351 & 0.008 & 0.757  & 1.116 \\
    2 & 2 & 4  & 96  & 228  & 0.675 & 0.007 & 1.531  & 2.213 \\
    3 & 4 & 9  & 239 & 749  & 0.741 & 0.009 & 3.227  & 3.977 \\
    4 & 4 & 9  & 445 & 1430 & 1.312 & 0.011 & 12.390 & 13.713 \\
    5 & 6 & 14 & 496 & 1540 & 1.427 & 0.003 & 15.471 & 16.901 \\
    6 & 6 & 14 & 886 & 3939 & 5.400 & 0.026 & 35.930 & 41.356 \\
    \hline
    \end{tabular}
    \vspace{-5pt}
\end{table*}

\subsection{Dynamic Environment without Precedence Specification}

Before evaluating the scalability of our approach, we first demonstrate its ability to handle dynamic environments without logical precedence constraints and illustrate the role of velocity constraints in determining the feasibility of a spacetime GCS problem. We consider a one-dimensional dynamic environment over the time horizon $t\in[0,17]$, where $x(t)\in\mathbb{R}$ denotes the position of the robot. The robot starts at $x(0)=0$ and is required to reach the target region $X_{\mathcal{T}}=[9,9.5]$ within the prescribed time horizon. Four spatial regions are inaccessible during specified time intervals:
\[
\begin{aligned}
X_1 &= [0.5,1.5], & \mathcal{I}_1 &= [0,2],\\
X_2 &= [2.5,3.5], & \mathcal{I}_2 &= [2,4],\\
X_3 &= [4.75,5.75], &
\mathcal{I}_3 &= [4,6]\cup[6.5,8.5]\cup[9.5,11.5],\\
X_4 &= [7,8], & \mathcal{I}_4 &= [15,17].
\end{aligned}
\]
where $\mathcal{I}_i$ denotes the time intervals during which $X_i$ is
inaccessible. Let $\mu_i := (x\in X_i)$,
the resulting STL specification is
\[
\varphi =
\left(
\bigwedge_{i=1}^{4}
\mathcal{G}_{\mathcal{I}_i}\neg\mu_i
\right)
\land
\mathcal{F}_{[0,17]}(X_{\mathcal{T}}),
\]
where $\mathcal{G}_{\mathcal{I}_i}$ denotes that the condition holds over all intervals contained in $\mathcal{I}_i$.

Three trajectories are shown in Fig.~\ref{fig:time_gated_plot} with different maximum velocity $v_{\text{max}}$ allowed. This example illustrates the interaction between dynamic obstacle availability and the robot's velocity limit. A larger velocity bound allows the robot to traverse narrow temporal openings before they close. As $v$ decreases, some of these openings become dynamically unreachable, forcing the trajectory to use later feasible passages. For a sufficiently restrictive velocity bound, the spatial connectivity visible in the spacetime environment no longer corresponds to a dynamically feasible continuous trajectory, and the specification cannot be satisfied within the prescribed time horizon, this result calls back to Remark~\ref{rem:dynamic_feasibility}.

\begin{figure}[t]
    \centering
    \includegraphics[width=0.6\linewidth]{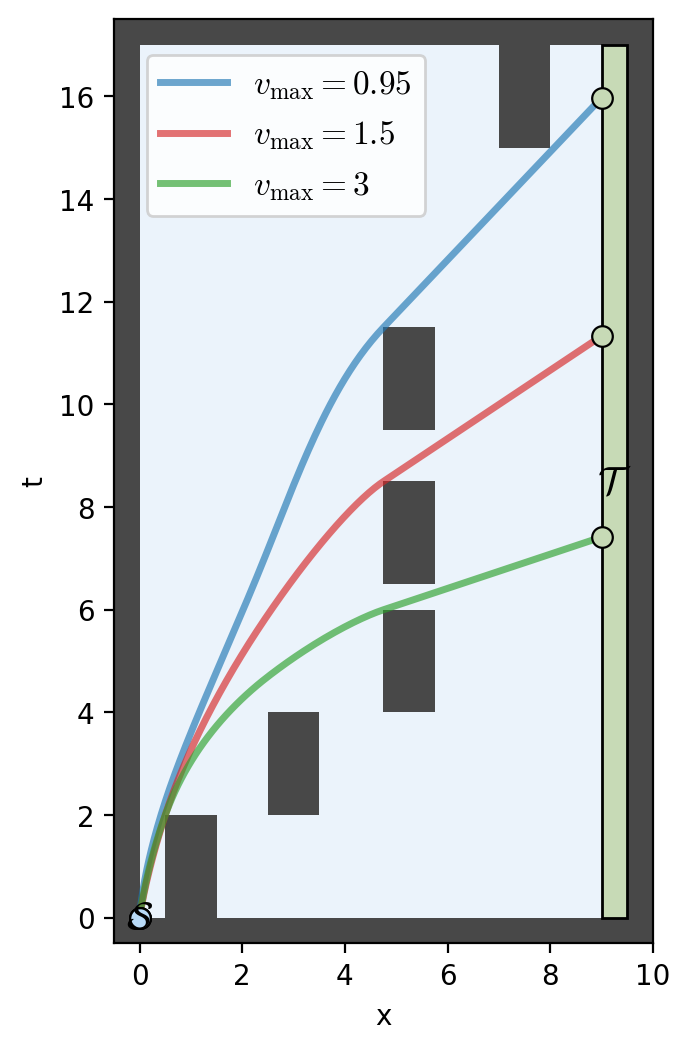}
    \caption{Dynamic obstacle avoidance in spacetime with one space dimension. The plot shows the quickest and shortest path to traverse on the $x$-axis with different maximum velocity constraint.}
    \label{fig:time_gated_plot}
    \vspace{-10pt}
\end{figure}

\subsection{Scalability of Spacetime Augmented GCS}

To evaluate the scalability of the proposed algorithm, we conduct three pairs of experiments, namely Experiments 1–2, 3–4, and 5–6, as summarized in Table~\ref{tab:scalability}. The two experiments within each pair share the same key–door configuration but differ in the complexity of obstacle motion. In the first experiment of each pair, each obstacle moves in a single direction throughout the entire time horizon, whereas in the second, each obstacle changes its direction of motion multiple times. The results show that increased geometric complexity leads to a larger number of partitioned free-space cells, while the combinatorial complexity of the augmentation remains unchanged, since the number of subgraphs is identical within each pair. This distinction is also reflected in the computation time: a larger number of free-space cells increases the time required for partitioning and constructing the spacetime GCS, whereas the time required to construct the augmented GCS remains approximately constant within each pair. Finally, the solution time increases with the overall size of the resulting augmented spacetime GCS, as characterized by its total number of vertices and edges.

\section{Conclusion}
\label{sec:conclusion}

This paper extended the augmented graphs of convex sets framework to motion planning with temporal-logic precedence specifications in dynamic environments. By lifting
piecewise-linearly translating and dilating polytopic regions into spacetime, we converted the time-varying geometry into a static polyhedral representation
in one additional dimension.  This construction allows the augmented-GCS subgraph lattice to retain its original logical structure while the continuous optimization accounts jointly for collision avoidance, causality, and bounded physical speed. It also naturally supports time-gated task regions and time-varying doors without introducing a separate temporal abstraction.

Several directions remain open. Improving scalability will require more efficient spacetime partitioning, implicit construction of the augmented graph, and stronger rounding or graph-search heuristics.  Extending the geometric model to rotating or nonpolytopic regions, incorporating uncertainty in obstacle motion, and developing receding-horizon variants for online replanning are also important topics. These extensions would broaden the applicability of spacetime augmented GCS while preserving its unified treatment of continuous motion, temporal evolution, and logical precedence.






\bibliographystyle{IEEEtran}
\bibliography{bibliography}

@article{marcucci2024shortest,
  author  = {Marcucci, Tobia and Umenberger, Jack and Parrilo, Pablo A. and Tedrake, Russ},
  title   = {Shortest Paths in Graphs of Convex Sets},
  journal = {SIAM Journal on Optimization},
  volume  = {34},
  number  = {1},
  pages   = {507--532},
  year    = {2024}
}

@article{marcucci2023motion,
  author  = {Marcucci, Tobia and Petersen, Mark and Von Wrangel, David and Tedrake, Russ},
  title   = {Motion Planning Around Obstacles with Convex Optimization},
  journal = {Science Robotics},
  volume  = {8},
  number  = {84},
  pages   = {eadf7843},
  year    = {2023}
}

@misc{marcucci2025unified,
  author       = {Marcucci, Tobia},
  title        = {A Unified and Scalable Method for Optimization Over Graphs of Convex Sets},
  howpublished = {arXiv preprint arXiv:2510.20184},
  year         = {2025}
}

@misc{you2025augmented,
  author       = {You, Shilin and Luna, Gabriel and Shaikh, Junaid and Gostin, Daniel and Xiang, Yu and Koeln, Justin and Summers, Tyler},
  title        = {Motion Planning with Precedence Specifications via Augmented Graphs of Convex Sets},
  howpublished = {arXiv preprint arXiv:2510.22015},
  year         = {2025}
}

@misc{you2026augmented,
  author = {You, Shilin and Luna, Gabriel and Summers, Tyler H.},
  title  = {A Framework for Motion Planning with Temporal Logic Precedence Specifications via Augmented Graphs of Convex Sets},
  howpublished   = {arXiv preprint arXiv:2606.00842},
  year   = {2026}
}

@article{kurtz2023temporal,
  author  = {Kurtz, Vince and Lin, Hai},
  title   = {Temporal Logic Motion Planning with Convex Optimization via Graphs of Convex Sets},
  journal = {IEEE Transactions on Robotics},
  volume  = {39},
  number  = {5},
  pages   = {3791--3804},
  year    = {2023}
}

@misc{morozov2024multiquery,
  author       = {Morozov, S. and Marcucci, Tobia and Amice, A. and Graesdal, B. P. and Bosworth, R. and Parrilo, Pablo A. and Tedrake, Russ},
  title        = {Multi-query Shortest-path Problem in Graphs of Convex Sets},
  howpublished = {arXiv preprint arXiv:2409.19543},
  year         = {2024}
}

@misc{graesdal2024contact,
  author       = {Graesdal, B. P. and Chia, S. Y. C. and Marcucci, Tobia and Morozov, S. and Amice, A. and Parrilo, Pablo A. and Tedrake, Russ},
  title        = {Towards Tight Convex Relaxations for Contact-rich Manipulation},
  howpublished = {arXiv preprint arXiv:2402.10312},
  year         = {2024}
}

@article{cohn2025noneuclidean,
  author  = {Cohn, Thomas and Petersen, Mark and Simchowitz, Max and Tedrake, Russ},
  title   = {Non-{E}uclidean Motion Planning with Graphs of Geodesically Convex Sets},
  journal = {The International Journal of Robotics Research},
  volume  = {44},
  number  = {10--11},
  pages   = {1840--1862},
  year    = {2025}
}

@misc{natarajan2024implicit,
  author       = {Natarajan, R. and Liu, C. and Choset, H. and Likhachev, M.},
  title        = {Implicit Graph Search for Planning on Graphs of Convex Sets},
  howpublished = {arXiv preprint arXiv:2410.08909},
  year         = {2024}
}

@misc{luna2026augmented,
  author       = {Luna, Gabriel and Summers, Tyler},
  title        = {Augmented Graphs of Convex Sets and the Traveling Salesman Problem},
  howpublished = {arXiv preprint arXiv:2604.06406},
  year         = {2026}
}

@article{bellman1962dynamic,
  author  = {Bellman, Richard},
  title   = {Dynamic Programming Treatment of the Travelling Salesman Problem},
  journal = {Journal of the ACM},
  volume  = {9},
  number  = {1},
  pages   = {61--63},
  year    = {1962}
}

@article{held1962dynamic,
  author  = {Held, Michael and Karp, Richard M.},
  title   = {A Dynamic Programming Approach to Sequencing Problems},
  journal = {Journal of the Society for Industrial and Applied Mathematics},
  volume  = {10},
  number  = {1},
  pages   = {196--210},
  year    = {1962}
}

@inproceedings{maler2004monitoring,
  author    = {Maler, Oded and Nickovic, Dejan},
  title     = {Monitoring Temporal Properties of Continuous Signals},
  booktitle = {International Symposium on Formal Techniques in Real-Time and Fault-Tolerant Systems},
  publisher = {Springer},
  pages     = {152--166},
  year      = {2004}
}

@inproceedings{tang2025space,
  title={Space-time graphs of convex sets for multi-robot motion planning},
  author={Tang, Jingtao and Mao, Zining and Yang, Lufan and Ma, Hang},
  booktitle={2025 IEEE/RSJ International Conference on Intelligent Robots and Systems (IROS)},
  pages={8683--8690},
  year={2025},
  organization={IEEE}
}

@article{osburn2025systematic,
  title={Systematic Constraint Formulation and Collision-Free Trajectory Planning Using Space-Time Graphs of Convex Sets},
  author={Osburn, Matthew D and Peterson, Cameron K and Salmon, John L},
  journal={arXiv preprint arXiv:2508.10203},
  year={2025}
}

@manual{mosek,
   author = "MOSEK",
   title = "The MOSEK Python Fusion API manual. Version 11.0.",
   year = 2025,
   url = "https://docs.mosek.com/latest/pythonfusion/index.html"
 }

@misc{drake,
 author = "Russ Tedrake and the Drake Development Team",
 title = "Drake: Model-based design and verification for robotics",
 year = 2019,
 url = "https://drake.mit.edu"
}

@article{de2026stl,
  title={STL-GCS: A Planner-Controller Framework for Signal Temporal Logic via Graphs of Time-varying Convex Sets},
  author={De Carli, Nicola and Marchesini, Gregorio and Dimarogonas, Dimos},
  journal={arXiv preprint arXiv:2607.19196},
  year={2026}
}

@article{avis1996reverse,
  title={Reverse search for enumeration},
  author={Avis, David and Fukuda, Komei},
  journal={Discrete applied mathematics},
  volume={65},
  number={1-3},
  pages={21--46},
  year={1996},
  publisher={Elsevier}
}

@article{tang2026unified,
  title={Unified Branch-and-Bound Search for the Steiner Traveling Salesman Problem on Graphs of Convex Sets},
  author={Tang, Jingtao and Ma, Hang},
  journal={arXiv preprint arXiv:2608.21319},
  year={2026}
}

@article{chen2026signal,
  title={Signal Temporal Logic Motion Planning via Graphs of Convex Sets},
  author={Chen, Yu and Hou, Ancheng and Feng, Mingyang and Yu, Xiao and Yin, Xiang},
  journal={arXiv preprint arXiv:2605.23240},
  year={2026}
}

@article{lin2025towards,
  title={Towards tighter convex relaxation of mixed-integer programs: Leveraging logic network flow for task and motion planning},
  author={Lin, Xuan and Ren, Jiming and Luo, Yandong and Xie, Weijun and Zhao, Ye},
  journal={arXiv preprint arXiv:2509.24235},
  year={2025}
}

@inproceedings{lin2025optimization,
  title={Optimization-based task and motion planning under signal temporal logic specifications using logic network flow},
  author={Lin, Xuan and Ren, Jiming and Coogan, Samuel and Zhao, Ye},
  booktitle={2025 IEEE International Conference on Robotics and Automation (ICRA)},
  pages={12565--12571},
  year={2025},
  organization={IEEE}
}

\end{document}